\documentclass{llncs}

\usepackage{mathtools}
\usepackage{todonotes}
\usepackage{tikz}
\usepackage{epstopdf}
\usepackage{amssymb}
\usepackage{nicefrac}
\usepackage[noend]{algorithm2e}
\usepackage{dirtytalk}

\usetikzlibrary{shapes,snakes,positioning,backgrounds,calc}

\newcommand{\calI}{\mathcal{I}}

\newcommand{\calT}{\mathcal{T}}
\newcommand{\calG}{\mathcal{G}}
\newcommand{\calV}{\mathcal{V}}
\newcommand{\calE}{\mathcal{E}}

\newcommand{\calH}{\mathcal{H}}
\newcommand{\calZ}{\mathcal{Z}}
\newcommand{\calN}{\mathcal{N}}

\newcommand{\uh}{u}

\makeatletter
\newcommand{\vast}{\bBigg@{4}}
\newcommand{\Vast}{\bBigg@{5}}
\makeatother

\newcommand{\braced}[1]{\left\lbrace #1 \right\rbrace}
\newcommand{\angled}[1]{\left\langle #1 \right\rangle}

\newcommand{\adj}{\mathrm{adj}}

\newcommand\numberthis{\addtocounter{equation}{1}\tag{\theequation}}

\begin{document}

\title{Dynamic Programming for One-Sided Partially Observable Pursuit-Evasion Games} 
\author{Karel Hor\'{a}k, Branislav Bo\v{s}ansk\'{y}\\
\texttt{\{horak,bosansky\}@agents.fel.cvut.cz}}
\institute{Department of Computer Science, Faculty of Electrical Engineering, \\Czech Technical University in Prague}

\maketitle
\begin{abstract}

Pursuit-evasion scenarios appear widely in robotics, security domains, and many other real-world situations. 
We focus on two-player pursuit-evasion games with concurrent moves, infinite horizon, and discounted rewards.
We assume that the players have a partial observability, however, the evader is given an advantage of knowing the current position of the units of the pursuer.
This setting is particularly interesting for security domains where a robust strategy, designed to maximize the utility in the worst-case scenario, is often desirable.
We provide, to the best of our knowledge, the first algorithm that provably converges to the value of a partially observable pursuit-evasion game with infinite horizon.
Our algorithm extends well-known value iteration algorithm by exploiting that (1) the value functions of our game depend only on position of the pursuer and the belief he has about the current position of the evader, and (2) that these functions are piecewise linear and convex in the belief space.
\end{abstract}

\section{Introduction}
Pursuit-evasion games appear in many scenarios in robotics and security domains~\cite{vidal2002probabilistic,chung2011search}, where a team of centrally controlled pursuing units (the \emph{pursuer}) aims to locate and capture the \emph{evader}, while the evader aims for the opposite.
We study this class of games and assume their discrete-time variant played on a finite graph.
We further assume that all units of both players move simultaneously, that the horizon of the game is infinite, the rewards are discounted over time with discount factor $\gamma \in [0,1)$, and that the players have only a partial information about the state of the world. 
Formally, such a game belongs to zero-sum partially observable stochastic games (POSGs).


We are interested in finding robust strategies of the pursuer against the worst-case evader.
Specifically, we assume that the evader knows the positions of the pursuing units and her only uncertainty is the strategy of the pursuer and the move that will be performed in the current time step.
Although in reality such perfectly informed adversary is rarely met, it is typical that the pursuer does not know what information is being revealed to the evader. Hence, in order to derive robust strategies (i.e. those that maximize pursuer's reward against \emph{any} type of the evader), it is natural to consider this kind of perfectly informed adversary.

We design the first algorithm that provably converges to the value of such one-sided partially observable pursuit-evasion games.
Moreover as the value converges, the strategies of the players converge to their optimal strategies as well.
This is in contrast to the existing approaches in robotics and security, where heuristic solutions without any optimality guarantees are used~\cite{vidal2002probabilistic,chung2011search}.


Our algorithm extends the well-known value iteration algorithm that is known to work for concurrent-moves stochastic games~\cite{shapley1953stochastic} as well as for partially observable models from decision theory -- Partially Observable Markov Decision Processes (POMDPs)~\cite{smallwood1973optimal,monahan1982state,pineau2003point,smith2012point}.
We adopt the methodology for POMDPs and show that one-sided pursuit-evasion games also allow us to define compactly represented value functions and thus to design a dynamic programming algorithm that iteratively improves values of game states over time and converges to the value of the game.
Specifically we show that the value functions (1) depend only on the position of the units of the pursuer and the belief he has about the possible position of the evader, but do \emph{not} depend on the history of moves, (2) these functions are piecewise linear and convex and thus can be, similarly to POMDPs, represented as a set of so called $\alpha$-vectors (Section~\ref{sec:shape}), and (3) we can design a dynamic-programming operator that provably converges to optimal value of the game when applied on these value functions (Section~\ref{sec:vi}).

We believe that our algorithm (accompanied with related theoretical results) is of comparable significance in this class of games as the full-backed value iteration algorithm in the class of POMDPs and it is a necessary first step towards designing practical scalable algorithms.

Due to the space constraints, most of the technical proofs can be found in the Appendix.

\subsection{Related Work}
A similar model with one-sided partial observability where one of the players has a perfect information was presented by McEneaney~\cite{mceneaney2004some}.
The author assumed that the player with perfect information knows the action the opponent plays at the current stage.
Due to the turn-based character of such game, the author considers only pure strategies.

The difference between our concurrent setup and the turn-based one is highlighted in Fig.~\ref{fig:mixed-required}. In the turn-based setting, the evader always observes the action played by the pursuer before making her action. This allows the evader to always move to a vertex that is unoccupied by the pursuer, hence preventing the pursuer from ever capturing him. On the other hand, in the concurrent setting, the only thing the evader can get to know is that the pursuer chooses every vertex with equal probability. She cannot anticipate the move of the pursuer and hence gets captured with probability 1/3 in the first round of the game.

Our setting with concurrent moves better corresponds to the real-world situations that occur in real time. 
As the evader does not know the action taken by the pursuer in the current stage, players may need to use randomized strategies.
However, allowing randomized strategies provides challenges in the design of the dynamic programming operator that we address in this paper.

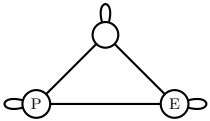
\begin{figure}[t]
\centering
\begin{tikzpicture}[,>=stealth,thick, scale=0.65, every node/.style={transform shape}, every loop/.style={}]
  \tikzstyle{circ} = [circle, draw, text width=15pt, text badly centered,
   node distance=2cm, inner sep=0pt]
  \node[circ] (topNode) {};
  \node[circ, below left of=topNode] (leftNode) {P};
  \node[circ, below right of=topNode] (rightNode) {E};
  \path[draw] (topNode)  -- (rightNode)
              (topNode)  -- (leftNode)
              (leftNode) -- (rightNode)
              (topNode)   edge[loop above] (topNode)
              (leftNode)  edge[loop left]  (leftNode)
              (rightNode) edge[loop right] (rightNode);
\end{tikzpicture}
\caption{A game with one unit of the pursuer (node marked with P) and the evader (node marked with E). Pure strategies are insufficient if players act simultaneously.}
\label{fig:mixed-required}
\end{figure}

Another model that uses one-sided partial observability was considered by Chatterjee et al.~\cite{chatterjee2014partial}, however with reachability and safety objectives (a player either wants to reach a set of target states or she wants to keep the system in a set of safe states) that do not translate to objectives with discounted rewards. 


An algorithm for solving a broader class of POSGs, where all players have imperfect information, was proposed in~\cite{hansen2004dynamic}, yet considering the finite horizon only.
The information is not shared among players, hence they typically attain different beliefs about the state of the game.
The algorithm uses ideas of dynamic programming to incrementally construct a set of pure strategies, longer ones from the shorter ones.
The set of such strategies is then used to form a normal-form representation of the POSG.
The use of iterated elimination of dominated strategies ensures that irrelevant strategies are not considered.
Nevertheless, in some games all strategies are relevant and this elimination does not help.
This is the case of the game in Fig.~\ref{fig:mixed-required} where the uniform play of the pursuer involves randomization over \emph{all} pure strategies considered by the algorithm;
the number of such strategies being exponential in the horizon.
The one-sided partial observability in our game allows us to avoid such enumeration of pure strategies, define dynamic programming over values of the subgames, and, in this particular case, represent the optimal \emph{infinite} horizon strategy using a single $\alpha$-vector.

\section{Finite-horizon game}
\label{sec:efg}
We use the notion of finite-horizon POSGs, or \emph{extensive-form games}, to reason about the infinite-horizon pursuit-evasion game with discounted rewards.
An~extensive-form game (EFG) is a tuple $G=(\calN,\calH,\calZ,\calT,u,\calI)$. $\calN$ is the set of players, in our case $\calN=\braced{p, e}$ where $p$ stands for the pursuer and $e$ for the evader.
Set $\calH$ denotes a finite set of \emph{histories} of actions taken by all players from the begining of the game.
Every history corresponds to a \emph{node} in the game tree; hence we use the terms history and node interchangeably.
Each of the histories may be either (1) \emph{terminal} ($h \in \calZ \subseteq \calH$) where the game ends and player $i$ gets utility $u_i(h)$, (2) controlled by the nature player who selects the successor node according to a fixed probability distribution known to all players, or (3) one of the players from $\calN$ may be to act.
We consider a zero-sum scenario where $u_p(h)=-u_e(h)$. To simplify the notation we use $u(h)$ to denote pursuer's reward.
An ordered list of transitions of player $i$ from root to node $h$ is referred to as a player $i$'s \emph{sequence}.
The allowed transitions in the game are modelled using a \emph{transition function} $\calT$ that provides a set of successor nodes for each non-terminal history.
The imperfect observation of players is modelled via \emph{information sets} $\calI_i$ that form a partition over histories $h$ where player $i \in \calN$ takes action.
We assume perfect recall setting where the players never forget their past actions, i.e. for every $I_i \in \calI_i$, all histories $h \in I_i$ have the same player $i$'s sequence.
Each information set $I_i \in \calI_i$ corresponds to one decision point of player $i$. 
A randomized \emph{behavioral strategy} of player $i$ assigns a distribution over actions to each of the information sets in~$\calI_i$.
A behavioral strategy of player $i$ can be represented in the form of a \emph{realization plan} $r$ which assigns probability of playing sequence $\sigma_i$ to each player $i$'s sequence $\sigma_i$. The behavioral strategy at information set $I_i \in \calI_i$ reached using a sequence $\sigma_i$ is then $b(I_i,a)=r(\sigma_i a) / r(\sigma_i)$.
A Nash equilibrium (NE) in an EFG is a pair of behavioral strategies, in which each player plays a best response to the strategy of the opponent. The expected utility of the pursuer when NE strategies are played by both players is referred to as the \emph{value of the game}.

We will now use this terminology to construct an EFG for a finite-horizon version of a pursuit-evasion game with $N$ pursuing units played on a graph $\calG=\angled{\calV,\calE}$ for $t$ rounds (we term $t$ as the \emph{horizon}). Part of the game tree is shown in Fig.~\ref{fig:efg}.
At every round $\tau \leq t$, pursuer's units occupy vertices $s_p^\tau$, where $s_p^\tau=\braced{s_{p,1}^\tau,\ldots,s_{p,N}^\tau}$ is an $N$-element multiset of vertices of $\calG$, and the evader is located in vertex $s_e^\tau \in \calV$.
The goal of the pursuer is to achieve a situation where the evader is caught, i.e. $s_e^\tau \in s_p^\tau$.
In every round, players have to move their units to vertices adjacent to their current positions ($\adj(v)$ denotes the set of vertices adjacent to $v$). Position of the evader in round $\tau+1$ is thus $s_e^{\tau+1} \in \adj(s_e^\tau)$.
We overload the operator $\adj$ to apply it also on multisets representing positions of pursuer's units, i.e. $s_p^{\tau+1} \in \adj(s_p^\tau)$, where $\adj(s_p^\tau) = \times_{i=1 \ldots N} \adj(s_{p,i}^\tau)$.

A horizon-$t$ game $G^t \!\!\angled{s_p^0,b^0}$ is parametrized by the initial position of the pursuer $s_p^0 \in \calV^N$ and a distribution over evader's initial positions $b^0 \in \Delta(\calV)$ known to both players (we term $b^0$ the initial \emph{belief}).
The game starts with a chance move selecting the initial position of the evader $s_e^0$ (based on $b^0$).

\begin{figure}[t]
\centering \includegraphics[scale=0.65]{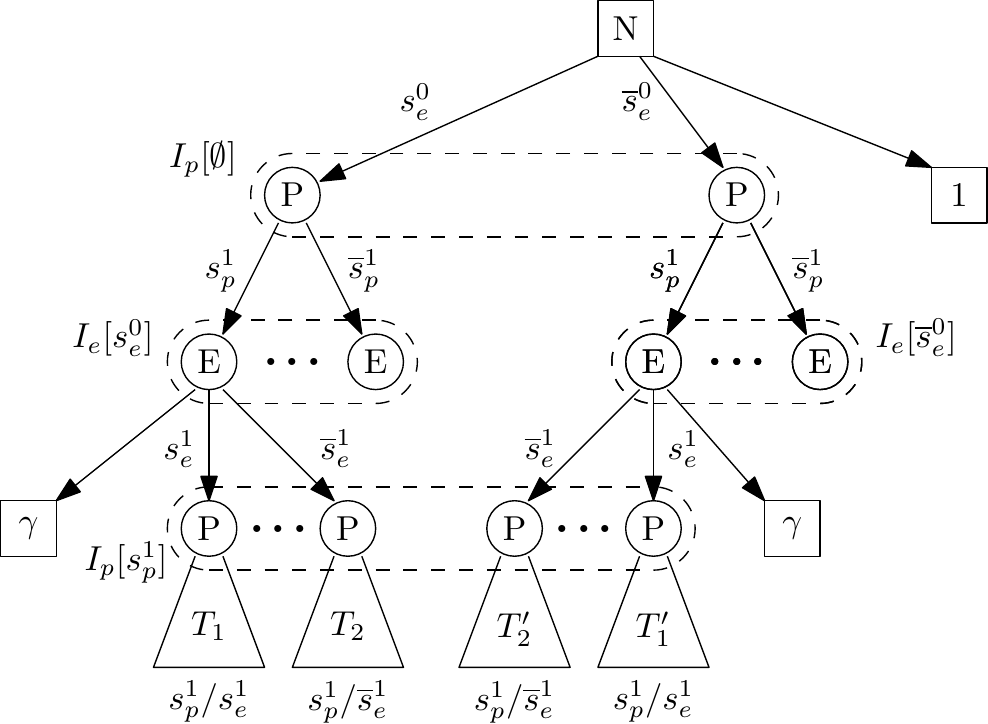}
\caption{An EFG representation of a finite-horizon pursuit-evasion game.}
\label{fig:efg}
\end{figure}

A history $h \in \calH$ in a game with horizon $t$ corresponds to a list of positions $s_e^0 s_p^1 s_e^1 \cdots s_p^\tau s_e^\tau$, where $\tau \leq t$.
The utility values are assigned to terminal histories as follows: in case the pursuer failed to capture the evader in time, i.e. if $\tau = t$ and $s_e^\tau \not\in s_p^\tau$, the pursuer gets utility $u(h)=0$; if he successfully captured the evader in the time limit $t$, i.e. if $\tau \leq t$ and $s_e^\tau \in s_p^\tau$, the pursuer gets the reward $u(h)=\gamma^\tau$ for capturing the evader in $\tau$ rounds (where $\gamma \in [0,1)$ stands for the \emph{discount factor}). 
The transition function $\calT$ complies with the graph (i.e., the adjacency function $\adj$), hence $s_p^\tau \in \adj(s_p^{\tau-1})$ and $s_e^\tau \in \adj(s_e^{\tau-1})$ for every $\tau \geq 1$. For notational simplicity we denote the sequence of pursuer's actions $s_p^1 \cdots s_p^\tau$ in $h$ as $h|_p$ and the sequence of evader's actions $s_e^1 \cdots s_e^\tau$ in $h$ as $h|_e$.

The position of the evader is unknown to the pursuer. Hence, in a perfect recall game, there is one pursuer's information set $I_p [\sigma_p]$ for each of his sequences $\sigma_p$ where $I_p[\sigma_p]=\braced{h' \;|\;\, h' \! \in \! \calH \! \setminus \! \calZ \! : h'|_p \! = \! \sigma_p}$.

Evader on the other hand knows the game situation almost perfectly. She knows where the pursuer's units were located \emph{before} the pursuer acted in the current round of the game (recall that we assume that pursuer acts first). The only information missing to the evader is the action being taken by the pursuer in the current round. Hence, for every history $h=s_e^0 s_p^1 s_e^1 \cdots s_p^\tau s_e^\tau$ where the pursuer is to play, there is evader's information set $I_e[h] = \braced{h s_p^{\tau+1} | s_p^{\tau+1} \in \adj(s_p^\tau)}$ containing all possible continuations of the pursuer.

\subsection{Shape of the value function}
\label{sec:shape}

The size of the extensive-form game representation and associated behavioral strategies grows exponentially as the horizon increases. This makes it quickly impossible to apply standard algorithms operating on game trees, especially since we aim to solve \emph{infinite} horizon games.

We aleviate the problem of increasing complexity of the strategy representation by representing them \emph{only} as their values.
We show that the value of a~strategy is linear in the belief, and we can thus represent it using just $|\calV|$ real numbers.
Moreover we show that there is a finite set of behavioral strategies that needs to be considered in a game with arbitrary finite horizon, regardless of the initial belief, and thus the value function representing the value of the best strategy at every belief is piecewise linear and convex, which allows us to represent this function in a compact manner.


\begin{definition}
A \emph{value function} $v^t\!\angled{s_p^0}: \Delta(\calV) \rightarrow [0,1]$ is a function assigning the value $v^t\!\angled{s_p^0}(b^0)$ of the game $G^t\!\angled{s_p^0,b^0}$ to every initial belief $b^0$ about the position of the evader.
By $v^t$ we mean a set of value functions $v^t\!\angled{s_p^0}$, one for each initial position $s_p^0 \in \calV^N$ of the pursuer.
\end{definition}

In the following text we show that a value function $v^t\!\angled{s_p^0}$ is piecewise linear and convex (PWLC) in the belief for every finite horizon $t$. For notational simplicity, the term linear will be used to refer to an affine function as well. The proof is structured as follows:
(1) first of all we show that the expected utility of every strategy of the pursuer is linear in the belief,
next (2) there is a finite set of behavioral strategies $\Sigma^t\!\angled{s_p^0}$ applicable in games $G^t\!\angled{s_p^0,\cdot\,}$ such that for every initial belief $b^0 \in \Delta(\calV)$ the pursuer has about the position of the evader, at least one strategy in $\Sigma^t\!\angled{s_p^0}$ is a NE solution of the game $G^t\!\angled{s_p^0,b^0}$; and
finally (3) we show that the PWLC nature of the value function follows from (1) and (2).

\begin{lemma}
\label{lemma:v:mixed-strategy-linear}
Let $\sigma_p$ be a randomized behavioral strategy of the pursuer in games $G^t\!\angled{s_p^0,b^0}$, where the pursuer starts in vertices $s_p^0$, parametrized by the initial belief about the position of the evader $b^0$. The expected utility of playing $\sigma_p$ against a best responding opponent is linear in $b^0$.
\end{lemma}
\begin{proof}
As the strategy $\sigma_p$ of the pursuer is fixed and the actual evader's position $s_e^0$ is revealed to the evader, it is not necessary from her perspective to consider the initial belief $b^0$ to derive the optimal evasive plan. She therefore chooses the optimal evasive plan (minimizing expected pursuer's utility) for every her possible initial position, value of each such plan being a constant. The initial belief $b^0$ forms a convex combination of values of individual evasive plans, which is a linear function in the belief space. \qed
\end{proof}

\begin{theorem}
\label{thm:finite-mixed-strategies}
Let $G^t\!\angled{s_p^0,b^0}$ be a horizon-$t$ game parametrized by the initial belief $b^0$ where the pursuer starts in a set of vertices $s_p^0$. There exists a finite set of pursuer's behavioral strategies $\Sigma^t \!\angled{s_p^0}$ such that for every initial belief $b^0$ about the position of the evader, the set $\Sigma^t \!\angled{s_p^0}$ contains at least one strategy $\sigma_p \in \Sigma^t \!\angled{s_p^0}$ that is in Nash equilibrium of $G^t\!\angled{s_p^0,b^0}$.
\end{theorem}
\begin{proof}
We use the sequence-form linear program for solving EFGs~\cite{koller1996efficient} to reason about the set of strategies $\Sigma^t \!\angled{s_p^0}$ the pursuer has to consider.
In this LP, values in every information set of the evader, as well as the value $v(root)$ in the root node of the game tree, are computed in a bottom-up fashion.
Every such value $v(I_e)$ of an information set $I_e$ can be seen as a concave piecewise linear function in the space of pursuer's realization plans (a compact representation of his behavioral strategies).
The pursuer then seeks for a realization plan that maximizes $v(root)$; the maximizer of which can be found among extreme points of line segments of $v(root)$, i.e. vertices of a polytope bounded by this function~\cite{vanderbei2014linear}.
We show that the set of such extreme points does not depend on the initial belief $b^0$.


There is one information set $I_e [ s_e^0 ]$ of the evader for each of her initial positions $s_e^0$.
The utility of every terminal node in the subgame beneath $I_e[s_e^0]$ is multiplied with chance probability $b(s_e^0)$, which allows us to factor out this probability and obtain the following constraint for the root node:
{\small
\begin{equation}
v(root) \leq \sum_{s_e^0 \in s_p^0} b^0(s_e^0) + \sum_{s_e^0 \in \calV \setminus s_p^0} b^0(s_e^0) \cdot \hat v(I_e [ s_e^0 ])
\end{equation}}

\vspace*{-1em}

Value $v(root)$ is a convex combination of concave piecewise linear functions $\hat v(I_e [ s_e^0 ])$.
As the belief was factored out, these functions, as well as the finite set of their extreme points $P[s_e^0]$, no longer depend on the belief.
This convex combination with arbitrary coefficients $b^0$ cannot have an extreme where none of the functions $\hat v(I_e [s_e^0])$ has one.
The set of extreme points is therefore a subset of $\bigcup_{s_e^0} P[s_e^0]$ --- a finite set that does not depend on the belief.
Each of the extreme points in $\bigcup_{s_e^0} P[s_e^0]$ corresponds to one pursuer's realization plan, and thus one his behavioral strategy, which allows us to construct the finite set $\Sigma^t \!\angled{s_p^0}$. \qed

\end{proof}

\begin{theorem}
\label{thm:convexity}
Value function $v^t\!\angled{s_p^0}$ is piecewise linear and convex in the belief space.
\end{theorem}
\begin{proof}
This result directly follows from Lemma~\ref{lemma:v:mixed-strategy-linear} and Theorem~\ref{thm:finite-mixed-strategies}. There is a finite set of randomized strategies $\Sigma^t \!\angled{s_p^0}$ that has to be considered by the pursuer and value of each such strategy is linear in the belief space. Thus the value function $v^t\!\angled{s_p^0}$ is a pointwise maximum taken over a finite set of linear functions, which is a PWLC function in the belief space. \qed
\end{proof}

Every PWLC function can be represented as a finite set of \emph{$\alpha$-vectors}. Every such $\alpha$-vector $\alpha=(\alpha_1,\ldots,\alpha_{|\calV|})$ represents one of the affine functions by assigning an expected reward $\alpha_i$ to each of the pure beliefs. We will often be working with the $\alpha$-vector representation of the value function, hence we overload the notation and consider value functions also as sets of such $\alpha$-vectors.

Lemma~\ref{lemma:v:mixed-strategy-linear} and Theorem~\ref{thm:finite-mixed-strategies} imply that every linear segment of the PWLC value function corresponds to one randomized strategy of the pursuer. This is similar to the POMDP case where each $\alpha$-vector corresponds to one conditional plan. This allows us to use terms $\alpha$-vector and pursuer's strategy interchangeably.

\section{Value iteration}
\label{sec:vi}
In the previous section, we related the concept of the value functions to the EFG representation of the game and proved that these functions have desirable properties (they are piecewise linear and convex).
We leverage their representation to design a dynamic programming approach inspired by value iteration algorithms for either POMDPs \cite{smallwood1973optimal,monahan1982state} or perfect information stochastic games \cite{shapley1953stochastic}.

The algorithm inductively constructs a sequence of value functions $\lbrace v^t \rbrace_{t=0}^\infty$, starting with values of a horizon-$0$ game (where the utility of the pursuer depends solely on the fact whether the evader starts in a vertex where one of the pursuer's units is located).

We avoid using the exponentially-sized representation of the underlying EFG by computing value function of a horizon-$t$ game using the solution of the game with horizon $t\!-\!1$.
First of all we show that there is a well-defined value update formula that expresses values of $v^t$ using value functions $v^{t-1}$ (Theorem~\ref{thm:value-update}).
We let the players choose their strategies for the first round of the horizon-$t$ game using the maximin principle (we term such strategies \emph{one-step strategies}) and we show that the pursuer can use these strategies to update his belief.
Pursuer's one-step strategy $\pi_p$ is a distribution over possible actions of his units, $\pi_p \in \adj(s_p^0)$, from which he samples his action.
The evader acts similarly, however she conditions her decision on her \emph{true} position $s_e^0$ (not just on the overall belief available to the pursuer); her one-step strategy is thus a mapping $\pi_e\!:\! \calV \rightarrow \Delta(\calV)$, such that $\pi_e(s_e^0)$ assigns zero probability to vertices not adjacent to $s_e^0$.

The piecewise linearity and convexity of value functions have implications on the computation of value functions. Firstly it allows us to find optimal one-step strategies by means of linear programming (Section~\ref{sec:lp}), furthermore it makes it possible to avoid evaluating the value update formula in every point in the belief space when constructing new value functions. Instead it is possible to construct new value functions incrementally and focus only on beliefs where the function being constructed has its extreme points of line segments (Section~\ref{sec:computing}). We conclude by showing that the dynamic programming operator used in the value iteration algorithm has a unique fixpoint corresponding to value functions of an infinite horizon game which ensures convergence properties of the algorithm (Theorem~\ref{thm:convergence}).

\begin{theorem}
\label{thm:value-update}
The value of the game $G^t \!\angled{s_p^0,b}$ can be computed from the solutions of horizon-$(t\!-1\!)$ games, whose values are represented by a set of value functions $v^{t-1}$. It holds that
\vspace*{-0.5em}
{\small
\begin{equation}
v^t\!\angled{s_p^0}(b) = \sum_{s_e \in s_p^0} \! b(s_e) + \gamma \left[ \sum_{s_e \in \calV \setminus s_p^0} \!\!\!\! b(s_e) \right] \cdot \max_{\pi_p} \min_{\pi_e} \! \sum_{s_p^1 \in \calV^N} \!\! \pi_p(s_p^1) \cdot v^{t-1}\!\angled{s_p^1}(b_{\pi_e})
\label{eq:value-update}
\end{equation}
}

\vspace*{-0.9em}

\noindent where the transformed belief $b_{\pi_e}$ depends solely on the evader's one-step strategy $\pi_e$ and the parametrization of the game $G^t \!\angled{s_p^0,b}$:
\vspace{-0.5em}
{\small
\begin{equation}
b_{\pi_e}(s_e') = \frac{1}{\sum_{s_e \in \calV \setminus s_p^0} b(s_e)} \sum_{s_e \in \calV \setminus s_p^0} \!\!\! b(s_e) \cdot \pi_e(s_e,s_e') \label{eq:belief}
\end{equation}
}

\vspace*{-1em}

\noindent The computation of $v^t$ by means of Equation~\eqref{eq:value-update} forms a dynamic programming operator $H$, such that $v^t=Hv^{t-1}$.
\end{theorem}
\begin{proof}
The correctness of the value update formula will be proven by computing the value of the game $G^t \!\angled{s_p^0,b}$ in a bottom-up fashion.
We start by considering that the one-step strategies of the players for the first round of the game are fixed, while they play optimally in the rest of the game.
This determines pursuer's expected reward at every node in the game tree, which allows us to express his expected utility in the root node of the game tree as an expectation over expected rewards in subsequent nodes (Lemma~\ref{lemma:root_value}).
Due to the fixed behavior in the first round of the game, parts of the game tree are independent on each other --- we refer to these subgames as $G[s_p^1]$.
This allows us to evaluate the expectation depicted in Lemma~\ref{lemma:root_value} by solving these games separately.
It turns out that games $G[s_p^1]$ are strategically equivalent to a game $G^{t-1}\!\angled{s_p^1,b_{\pi_e}}$ with shorter horizon, the solution of which is represented by value functions $v^{t-1}$.
The expectation can thus be expressed solely in terms of $v^{t-1}$.
Finally the assumption of fixed one-step strategies for the first round of the game gets relaxed, which yields the desired maximin formula from Equation~\eqref{eq:value-update}.
The derivation of the value update formula relies on several technical lemmas, the proofs of which can be found in the Appendix.

Let $\pi_p \!\in\! \Delta(\adj(s_p^0))$ be a fixed pursuer's one-step strategy, and $\pi_e : \calV \rightarrow \Delta(\calV)$ be a fixed one-step strategy of the evader. Assume that both players play according to $\pi_p$ and $\pi_e$ in the first round of the game, i.e. the pursuer follows $\pi_p$ in his information set $I_p[\emptyset]$ (i.e. pursuer's information set where he has not acted yet, see Fig.~\ref{fig:efg}) and the evader plays according to $\pi_e(s_e^0)$ in her information set $I_e[ s_e^0 ]$ (where she has received the information that she is located in vertex $s_e^0$). Once the first round of the game is over, players continue with their best strategies for the situation they are currently in. We denote such optimal strategies where the players are restricted to play $\pi_p$ and $\pi_e$ in the first round as $\sigma_p$ and $\sigma_e$.


\begin{definition}
Let $\pi_p$ and $\pi_e$ be fixed one-step strategies of the players for the first round of the game $G^t \!\angled{s_p^0,b}$ and let $\sigma_p$, $\sigma_e$ be optimal strategies of the players with the restriction to play $\pi_p$ and $\pi_e$ in the first round. The expected reward of the pursuer when strategies $(\sigma_p,\sigma_e)$ are followed and node $h$ in the game tree is reached is denoted $\uh(h)$ and termed \emph{expected reward in $h$}.
\end{definition}

We follow by expressing the expected utility the pursuer gets when strategies $(\sigma_p,\sigma_e)$ are followed by propagating expected rewards from subsequent nodes in the game tree.
We use histories of the form $s_e^0 s_p^1 s_e^1$ where the evader started in vertex $s_e^0$ (based on the move of nature) and then, during the first round of the game, the pursuer moved his units to vertices $s_p^1$ and the evader moved to $s_e^1$.

\begin{lemma}
\label{lemma:root_value}
The expected reward in the root node equals to:
\vspace*{-0.2em}
{\small
\begin{align*}
\uh(\emptyset) &= \sum_{s_e^0 \in s_p^0} b(s_e^0) + \left[ \sum_{s_e^0 \not\in s_p^0} b(s_e^0) \right] \cdot \sum_{s_p^1} \pi_p(s_p^1) \Bigg( \gamma \sum_{s_e^1 \in s_p^1} b_{\pi_e}(s_e^1) + \\
& \hspace{-30pt} + \Bigg[ \sum\limits_{s_e^1 \not\in s_p^1} b_{\pi_e}(s_e^1) \Bigg] \sum_{s_e^1 \not\in s_p^1} \sum\limits_{s_e^0 \not\in s_p^0} \Bigg[ \frac{b(s_e^0) \cdot \pi_e(s_e^0,s_e^1)}{\sum\limits_{\tilde{s}_e^1 \not\in s_p^1} \sum\limits_{\tilde{s}_e^0 \not\in s_p^0} b(\tilde{s}_e^0) \cdot \pi_e(\tilde{s}_e^0,\tilde{s}_e^1)} \cdot \uh(s_e^0 s_p^1 s_e^1) \Bigg] \Bigg) \numberthis\label{eq:root_value}
\end{align*}
}
\end{lemma}

Lemma~\ref{lemma:root_value} expressed the value in the root node based on the expected rewards in histories $s_e^0 s_p^1 s_e^1$ where the pursuer is to move.
The pursuer knows only $s_p^1$, hence these histories are partitioned into his information sets $I_p[s_p^1]$, one for each move $s_p^1$ of the pursuer in the first round (see Fig.~\ref{fig:efg}).
Importantly, for every subgame below $I_p[s_p^1]$, there is no information set that would involve nodes not present in this subgame --- neither pursuer nor evader forget that $s_p^1$ was played.
The optimal behavior in these subgames therefore depends only on the belief in $I_p[s_p^1]$, which is fixed due to the fixed behavior in the first round.
We can therefore compute value of the subgame below $I_p[s_p^1]$ separately by making chance simulate the belief in this information set.



Let us construct a game $G[s_p^1]$ which consists of the information set $I_p[s_p^1]$ and the subgame beneath it.
In this game, information set $I_p[s_p^1]$ is reached with probability $\beta = \sum_{s_e^1 \not\in s_p^1} b_{\pi_e}(s_e^1)$, while with probability $1-\beta$ the pursuer gets utility $\gamma$ without play --- this accounts for the reward the pursuer gets if he catches the evader in the first round by playing action $s_p^1$.
The nature player simulates the belief $b[s_p^1]$ in the information set $I_p[s_p^1]$, so that the probability of every history in this information set, given this information set was reached, is identical with the original game.
The value of the game $G[s_p^1]$ corresponds to the following part of the Equation~\eqref{eq:root_value}:
\vspace*{-0.5em}
{\small
\begin{equation}
\gamma \!\!\!\! \underbrace{\sum_{s_e^1 \in s_p^1} \!\! b_{\pi_e}(s_e^1)}_{\substack{\text{Evader caught}\\\text{in the first round}}} \!\!\! + \!\underbrace{\left[ \sum\limits_{s_e^1 \not\in s_p^1} \!\! b_{\pi_e}(s_e^1) \right]}_{\substack{\text{Evader not caught}\\\text{in the first round}}} \sum_{s_e^1 \not\in s_p^1} \sum\limits_{s_e^0 \not\in s_p^0} \!\vast[ \!\!\!\!\!\!\!\underbrace{\frac{b(s_e^0) \cdot \pi_e(s_e^0,s_e^1)}{\sum\limits_{\tilde{s}_e^1 \not\in s_p^1} \sum\limits_{\tilde{s}_e^0 \not\in s_p^0} \!\! b(\tilde{s}_e^0) \cdot \pi_e(\tilde{s}_e^0,\tilde{s}_e^1)}}_{\substack{\text{Belief $b[s_p^1]$ of history $s_e^0 s_p^1 s_e^1$ in $I_p[s_p^1]$}}} \!\!\!\!\!\!\!\,\cdot\, \uh(s_e^0 s_p^1 s_e^1) \vast]
\label{eq:infset-game}
\end{equation}}

In the case of the game $G[s_p^1]$, there are multiple histories for every \emph{current} position of the evader $s_e^1$ in the information set $I_p[s_p^1]$ (resulting from different initial locations of the evader $s_e^0$).
We show that it is not necessary to account for different initial positions of the evader $s_e^0$, and thus all histories in $I_p[s_p^1]$ having the same current position of the evader $s_e^1$ can be merged.
The resulting game contains a single history for each $s_e^1$ in $I_p[s_p^1]$, and thus this game is equivalent to a shorter horizon game $G^{t-1}\!\angled{s_p^1,b_{\pi_e}}$ up to multiplication of the utilities by $\gamma$ to account for a round that has already passed. This allows using the solution of $G^{t-1}\!\angled{s_p^1,b_{\pi_e}}$ represented by value functions $v^{t-1}$ to express the value of $G[s_p^1]$.

\begin{definition}
Two deterministic game trees over nodes $\calH_1,\calH_2$ are isomorphic if there exists a bijection $\xi: \calH_1 \rightarrow \calH_2$ such that $v \in \calH_1$ is a successor of $u \in \calH_1$ if and only if $\xi(v)$ is a successor of $\xi(u)$, $n \in \calH_1$ is a pursuer's node if and only if $\xi(n)$ is a pursuer's node, it is a terminal node if and only if $\xi(n)$ is a terminal node and the utilities $u(n) = u(\xi(n))$. Moreover the trees have the same informational structure so that two nodes $u,v \in \calH_1$ are in the same information set if and only if nodes $\xi(u),\xi(v)$ are in the same information set.
\end{definition}

We can observe that subtrees of nodes $s_e^0 s_p^1 s_e^1$ and $\overline{s}_e^0 s_p^1 s_e^1$ (where $s_e^0$ and $\overline{s}_e^0$ stands for two different initial positions of the evader) are isomorphic as we can establish a bijection $\xi(s_e^0 s_p^1 s_e^1 h_{rest})=\overline{s}_e^0 s_p^1 s_e^1 h_{rest}$. The utility of terminal histories does not depend on the initial position of the evader (only on the time when the pursuer managed to capture the evader). Whenever pursuer's node $u$ is in information set $I_p$, node $\xi(u)$ is in $I_p$ as well (because pursuer has no way to detect the evader's initial position).
Moreover whenever evader cannot distinguish between two histories $s_e^0 s_p^1 s_e^1 \cdots s_p^q$ and $s_e^0 s_p^1 s_e^1 \cdots \overline{s}_p^q$, she cannot distinguish between histories $\overline{s}_e^0 s_p^1 s_e^1 \cdots s_p^q$ and $\overline{s}_e^0 s_p^1 s_e^1 \cdots \overline{s}_p^q$ either (because her uncertainty is related to the pursuer's move at round $q$, which does not depend on the initial position of the evader).
Thus the subtrees have also the same informational structure.

\begin{lemma}
\label{lemma:subtree-elimination}
Let $I$ be the topmost information set of game $G[s_p^1]$ and let the belief $b[I]$ over nodes from $I$ be known and fixed. Let $n_1,n_2 \in I$ be two nodes whose subtrees are isomorphic. Then a game $G'$ with the same structure as $G$ with any belief $b'[I]$ in $I$, satisfying $b[n_1]+b[n_2]=b'[n_1]+b'[n_2]$ and $b[n]=b'[n]$ for all nodes other than $n_1$ and $n_2$, has the same value as $G$.
\end{lemma}

Thanks to the Lemma~\ref{lemma:subtree-elimination} and the isomorphism of the subtrees beneath $s_e^0 s_p^1 s_e^1$ and $\overline{s}_e^0 s_p^1 s_e^1$, histories $s_e^0 s_p^1 s_e^1$ and $\overline{s}_e^0 s_p^1 s_e^1$ can be merged and associated beliefs added up. By repeating this process, we end up with a single history for each current position of the evader $s_e^1$ (let $s_e^0 s_p^1 s_e^1$ be such history), whose belief is
\vspace*{-0.5em}
{\small
\begin{equation}
b'[s_p^1](s_e^0 s_p^1 s_e^1) \coloneqq \frac{\sum\limits_{s_e^0 \not\in s_p^0} b(s_e^0) \cdot \pi_e(s_e^0,s_e^1)}{\sum\limits_{\tilde{s}_e^1 \not\in s_p^1} \sum\limits_{\tilde{s}_e^0 \not\in s_p^0} b(\tilde{s}_e^0) \cdot \pi_e(\tilde{s}_e^0,\tilde{s}_e^1)} = \frac{b_{\pi_e}(s_e^1)}{\sum\limits_{\tilde{s}_e^1 \not\in s_p^1} b_{\pi_e}(s_e^1)}; \;\; b'[s_p^1](s_e^1) \text{ for short}
\label{eq:is-belief}
\end{equation}}

The updated belief $b'[s_p^1]$ from Equation~\eqref{eq:is-belief} complies with belief $b_{\pi_e}$ (computed according to Equation~\eqref{eq:belief}) updated with the information that the evader is located in none of the vertices of $s_p^1$. The belief in $I_p[s_p^1]$ is identical with the belief in top-level information set of the game $G^{t-1}\!\angled{s_p^1,b_{\pi_e}}$; and hence the resulting game is identical to the game $G^{t-1}\!\angled{s_p^1,b_{\pi_e}}$ up to the multiplication by $\gamma$. The value of the game $G[s_p^1]$ (Equation~\eqref{eq:infset-game}), from which this game was derived, is thus $\gamma v^{t-1}\!\angled{s_p^1}(b_{\pi_e})$. We substitute this value to Equation~\eqref{eq:root_value} to obtain
\vspace*{-0.4em}
{\small
\begin{align*}
\uh(\emptyset)
&= \sum_{s_e^0 \in s_p^0} \! b(s_e^0) + \left[ \sum_{s_e^0 \not\in s_p^0} \! b(s_e^0) \right] \cdot \sum_{s_p^1} \pi_p(s_p^1) \cdot \Bigg( \gamma v^t\angled{s_p^1}(b_{\pi_e}) \Bigg) \numberthis\label{eq:freezed}
\end{align*}
}

\vspace*{-0.5em}

\noindent By allowing the players to choose their optimal one-step strategies $\pi_p$ and $\pi_e$ in Equation~\eqref{eq:freezed}, we obtain the desired maximin formula shown in Equation~\eqref{eq:value-update}. \qed

\end{proof}

\subsection{From value functions to optimal one-step strategies}
\label{sec:lp}
The evaluation of the maximin formula from Equation~\eqref{eq:value-update} involves computation of optimal strategies of the players. In this section we show that if the value functions $v^{t-1}$ are piecewise linear and convex functions represented by sets of $\alpha$-vectors (which holds due to Theorem~\ref{thm:convexity}), the strategies can be found out by means of linear programming.

Due to the space constraints, we provide only the linear program for computing optimal one-step strategy of the pursuer in the game $G^t\!\angled{s_p^0,b}$ (its dual for computing evader's strategy can be found in the Appendix). At the beginning of each round, the pursuer realizes what vertices the evader is \emph{not} located in, and hence updates his belief about the position of the evader. We can therefore restrict ourselves to the case where $b(s_e)=0$ for all $s_e \in s_p^0$.



In the following linear program, the pursuer seeks for a strategy maximizing his expected utility against the best-responding opponent. He assumes strategies of the form \say{move to $s_p^1$ first and then follow strategy whose value is represented by $\alpha \in v^{t-1}\!\angled{s_p^1}$}. The choice of $\alpha$ uniquely defines such strategy.
The probability of playing each strategy $\alpha \in v^{t-1}\!\angled{s_p^1}$ is represented by variable $\hat\pi_p(s_p^1,\alpha)$. Constraint~\eqref{lp:pursuer:br} corresponds to the value of playing such randomized strategy against the best-responding evader who starts in vertex $s_e$ ($\alpha(s_e')$ denotes the value of $\alpha$ evaluated at pure belief corresponding to action $s_e'$ of the evader). The evader starts in $s_e$ with probability $b(s_e)$, hence the objective~\eqref{lp:pursuer} calculates the expectation over individual $v(s_e)$.
For the resulting one-step strategy of the pursuer, it holds that $\pi(s_p^1) = \sum_{\alpha \in v^{t-1}\!\angled{s_p^1}} \hat\pi(s_p^1,\alpha)$.
{\small
\begin{align}
\max_{v,\hat{\pi}_p}\  & \gamma \sum_{s_e \in \calV} b(s_e) \cdot v(s_e) \label{lp:pursuer}\\
\text{s.t.} \  & \!\! \sum_{s_p^1 \in \adj(s_p^0) \;;\; \alpha \in v^{t-1}\angled{s_p^1}} \!\!\!\!\!\!\!\!\!\!\!\!\!\!\!\!\!\!\!\!\!\alpha(s_e') \cdot \hat{\pi}_p(s_p^1,\alpha) \geq v(s_e) && \forall s_e \in \calV \ \forall s_e' \in \adj(s_e) \label{lp:pursuer:br} \\
& \!\! \sum_{s_p^1 \in \adj(s_p^0) \;;\; \alpha \in v^{t-1}\angled{s_p^1}} \!\!\!\!\!\!\!\!\!\!\!\!\!\!\!\!\!\!\!\!\! \hat{\pi}_p(s_p',\alpha) = 1 \label{lp:dual:probability} \\
& \;\;\;\;\;\;\;\;\;\;\;\;\;\; \hat{\pi}_p(s_p^1,\alpha) \geq 0 && \forall s_p^1 \in \adj(s_p^0) \; \forall \alpha \in v^{t-1}\angled{s_p^1}
\end{align}
}


\subsection{Computing value functions}
\label{sec:computing}

In each iteration of our value iteration algorithm, value functions $v^t$ are constructed from the solution from the previous iteration --- value functions $v^{t-1}$.
By repeating this construction, a sequence of finite-horizon value functions $\braced{v^t}_{t=0}^\infty$ approaching the values of the infinite-horizon game is being constructed.
The value functions $v^t$ that are about to be constructed, as well as $v^t$, are piecewise linear and convex (Theorem~\ref{thm:convexity}).
In this section, we show that this allows us to avoid evaluating the dynamic programming operator $H$ (Equation~\eqref{eq:value-update}) in every point in the belief space and enables us to construct $v^t$ by considering only a finite subset of beliefs, corresponding to the extreme points of line segments of $v^t$.
We proceed in two steps: (1) firstly we compute a function $Q^t_{\pi_p}\!\angled{s_p^0}$ corresponding to the expected utility the pursuer gets if he plays $\pi_p$ at the first round of the longer horizon game $G^t\!\angled{s_p^0,b}$; (2) then we show how to compute $v^t\!\angled{s_p^0}$ as a combination of multiple $Q^t_{\pi_p}\!\angled{s_p^0}$ for properly chosen one-step strategies $\pi_p$.
We start with a formal definition of function $Q^t_{\pi_p}\!\angled{s_p^0}$.

\begin{definition}
Let $\pi_p$ be pursuer's one-step strategy for the first round of the game $G^t\!\angled{s_p^0,b}$. The \emph{value of $\pi_p$} is a function $Q^t_{\pi_p}\!\angled{s_p^0}$ assigning the expected reward the pursuer gets in the game $G^t\!\angled{s_p^0,b}$ against the best-responding opponent, when he plays $\pi_p$ in the first round and continues by playing according to his optimal strategy in the rest of the game, i.e.
\vspace*{-0.5em}
{\small
\begin{equation}
Q^t_{\pi_p}\!\angled{s_p^0}(b) \coloneqq \sum_{s_e \in s_p^0} \!\! b(s_e) + \gamma \left[ \sum_{s_e \in \calV \setminus s_p^0} \!\!\!\! b(s_e) \right] \cdot \min_{\pi_e} \! \sum_{s_p^1 \in \calV^N} \!\!\! \pi_p(s_p^1) \cdot v^{t-1}\!\angled{s_p^1}(b_{\pi_e})
\end{equation}
}
\end{definition} 

According to the previous definition, once the first round of the game is over, the pursuer continues with his optimal strategy.
The following lemma shows that this optimal strategy for the rest of the game is well-defined and its value is characterized by the value functions $v^{t-1}$.

\begin{lemma}
\label{lemma:Qt}
Let $\pi_p$ be pursuer's fixed one-step strategy for the first round of the game. For every belief $b$ there are strategies $\sigma_p[s_p^1]$, one for each $s_p^1 \in \adj(s_p^0)$, represented by $\alpha$-vectors $\alpha[s_p^1] \in v^{t-1}\!\angled{s_p^1}$, such that it is optimal to follow $\sigma_p[s_p^1]$ when $s_p^1$ was played in the first round of the game. The value of strategy $\sigma_p$ prescribing the pursuer to play according to $\pi_p$ in the first round and continue by using respective $\sigma_p[s_p^1]$ is linear and the corresponding $\alpha$-vector satisfies
\vspace*{-0.5em}
{\small
\begin{equation}
\alpha^{\sigma_p}(s_e) = \begin{cases}
1 & s_e \in s_p \\
\gamma \min\limits_{s_e' \in \adj(s_e)} \sum\limits_{s_p^1} \pi_p(s_p^1) \cdot \alpha[s_p^1](s_e') & \mathrm{otherwise}
\end{cases}
\end{equation}}
\end{lemma}


Lemma~\ref{lemma:Qt} gives us a direct algorithm for computing $Q^t_{\pi_p}$.
PWLC functions $v^{t-1}$ correspond to a finite number of horizon-$t$ strategies, represented by a finite number of $\alpha$-vectors.
Thus there is only a finite number of ways to choose strategies $\sigma_p[s_p^1]$ from Lemma~\ref{lemma:Qt}, which can be found by means of enumeration.
The maximization over linear functions representing value of such strategies corresponds to the function $Q^t_{\pi_p}\!\angled{s_p}$ which is thus piecewise linear and convex.

The definition of $Q^{t}_{\pi_p}\!\angled{s_p}$ implies that we can compute the value function $v^{t+1}\!\angled{s_p}$ by allowing the pursuer to play arbitrary strategy $\pi_p$, when
\vspace*{-0.5em}
{\small
\begin{equation}
v^t\!\angled{s_p^0}(b) = \max_{\pi_p} Q^t_{\pi_p}\!\angled{s_p^0}(b)
\label{eq:vtqt}
\end{equation}
}

\vspace*{-0.5em}

As a consequence of Theorem~\ref{thm:finite-mixed-strategies}, it is sufficient to consider a finite set $\Pi_p$ of strategies in the maximizer of Equation~\eqref{eq:vtqt} and thus obtain $v^t\!\angled{s_p^0}$ as the pointwise maximum from respective $Q^t_{\pi_p}\!\angled{s_p^0}$ functions, $v^t\!\angled{s_p} = \bigoplus_{\pi_p \in \Pi_p} Q^t_{\pi_p}\!\angled{s_p}$. The set of such strategies $\Pi_p$ is however initially unknown. We propose an algorithm (Algorithm~\ref{alg:backup}) that constructs both the set of strategies $\hat\Pi_p$ and the value function $\hat v^t\!\angled{s_p^0}$ incrementally by iteratively verifying whether the current set of the strategies $\hat\Pi_p$ is sufficient for obtaining the actual value function $v^t\!\angled{s_p^0}$.

\begin{algorithm}[t]
\DontPrintSemicolon
$\hat{v}^t\!\angled{s_p^0} \gets \braced{\mathbf{0}^{|\calV|}}$ \;
$\hat{\Pi}_p = \emptyset$ \;
\While{$\exists b \in \Delta(\calV), \pi_p \not\in \Pi_p : Q^t_{\pi_p}\!\angled{s_p^0}(b) > \hat{v}^t\!\angled{s_p^0}(b)$}{
  $\pi_p \gets $ optimal strategy of the pursuer at belief $b$ for the first round (see~\eqref{lp:pursuer}) \;
  $\hat{\Pi}_p \gets \hat{\Pi}_p \cup \braced{\pi_p}$ \;
  $\hat{v}^t\!\angled{s_p^0} \gets \hat{v}^t\!\angled{s_p^0} \oplus Q^t_{\pi_p}\!\angled{s_p}$ \;
}
\Return{$\hat{v}^t\!\angled{s_p}$}
\caption{Incremental construction of value function $v^t\!\angled{s_p}$}
\label{alg:backup}
\end{algorithm}

The Algorithm~\ref{alg:backup} is constructing a set of strategies $\hat{\Pi}_p$ and a corresponding estimate of value function $\hat{v}^t\!\angled{s_p^0}=\bigoplus_{\pi_p \in \hat{\Pi}_p} Q^t_{\pi_p}\!\angled{s_p^0}$, starting with an empty set $\hat{\Pi}_p$.
At each iteration, it verifies whether strategies $\hat\Pi_p$ used to compute current $\hat{v}^{t+1}\angled{s_p^0}$ are optimal in every belief $b \in \Delta(\calV)$. 
If it finds a belief $b$ where the strategy can be improved, i.e. there exists $\pi_p$ such that $Q^t_{\pi_p}\!\angled{s_p^0}(b) > \hat{v}^t\!\angled{s_p^0}(b)$, it updates the set $\hat{\Pi}_p$ and recomputes $\hat{v}^t\!\angled{s_p}$.
If no such belief is found, all required strategies were considered and $\hat{v}^t\!\angled{s_p^0}=v^t\!\angled{s_p^0}$.

Whenever the value function $\hat{v}^t\!\angled{s_p^0}$ is not yet optimal in the whole belief space, i.e. there exists a belief $b$ where $Q^t_{\pi_p}\!\angled{s_p^0}(b) > \hat{v}^t\!\angled{s_p^0}(b)$, there exists a belief $b'$ with the same property that forms an extreme point of a line segment on $\hat{v}^t\!\angled{s_p^0}$.
This property is characterized by Lemma~\ref{lemma:belief-finite}.

\begin{lemma}
\label{lemma:belief-finite}
If there is a belief $b$ where $v^t\!\angled{s_p^0}(b) > \hat{v}^t\!\angled{s_p^0}(b)$, there must be a belief $b'$ that forms an extreme point of a line segment on the surface of $\hat{v}^t\!\angled{s_p^0}$ where $v^t\!\angled{s_p^0}(b') > \hat{v}^t\!\angled{s_p^0}(b')$.
\end{lemma}

Thanks to Lemma~\ref{lemma:belief-finite}, we can consider only a \emph{finite} set of beliefs that form extreme points of line segments on the value function $\hat{v}^t\!\angled{s_p^0}$.
In every iteration, a one-step strategy that is optimal at some belief point (and thus must be present in $\Pi_p$) is added to the set $\hat\Pi_p$.
As a consequence of the Theorem~\ref{thm:finite-mixed-strategies}, the set $\Pi_p$ that is necessary to obtain the optimal value function $v^t\!\angled{s_p^0}$ is finite.
Hence after a finite number of iterations, the Algorithm~\ref{alg:backup} terminates.

\subsection{Convergence and uniqueness of the solution}
\label{sec:convergence}
We demonstrate the convergence properties of our value iteration algorithm and the uniqueness of the value functions solving the infinite horizon game by showing that the dynamic programming operator $H$ (Equation~\ref{eq:value-update}) is a contraction mapping. The desired properties then follows from the Banach's fixed point theorem~\cite{ciesielski2007stefan}. We show the contractivity of $H$ under the following max-norm:
\vspace*{-0.1em}
{\small
\begin{equation}
\| v - \overline{v} \| = \max_{s_p^0 \in \calV^N} \max_{b \in \Delta(\calV)} | v\angled{s_p^0}(b) - \overline{v}\angled{s_p^0}(b) |
\end{equation}}

\vspace*{-0.5em}

\begin{lemma}
\label{lemma:contractivity}
The operator $H$ is a contraction with contractivity factor $\gamma < 1$ under max-norm.
\end{lemma}

\begin{theorem}
\label{thm:convergence}
There is a unique set of value functions $v^*$ satisfying $v^* = Hv^*$ and the recursive application of $H$ converges to $v^*$. Series $\braced{v^t}_{i=0}^{\infty}$ thus converges to value functions of an infinite horizon game.
\end{theorem}
\begin{proof}
The operator $H$ is a contraction mapping defined on a metric space of sets of bounded functions defined on the belief space. By applying Banach's fixed point theorem \cite{ciesielski2007stefan} we get that $H$ has a unique fixed point $v^*$ and the recursive application of $H$ converges to $v^*$.
\qed
\end{proof}

\begin{proposition}
\label{prop:convergence}
After $t$ iterations of the value iteration algorithm, the value function $v^t$ is $\gamma^t$-optimal (i.e. $\| v^t - v^* \| \leq \gamma^t$).
\end{proposition}

\section{Conclusion}
We present the first algorithm for solving the class of two-player discounted pursuit-evasion games with infinite horizon and partial observability, where the evader is assumed to be perfectly informed about the current state of the game (i.e. position of pursuer's units). 
This class of games has a significant relevance in security domains where a robust strategy that provides guarantees in the worst case is often desirable.

Our algorithm is a modification of the well-known value iteration algorithm for solving Partially Observable Markov Decision Processes (POMDPs), or stochastic games with concurrent moves.
We show that the strategies can be compactly represented using value functions that depend on the location of the pursuing units and the belief about the position of the evader, but not explicitly on the history of moves. 
These value functions are piecewise linear and convex and allow us to design a dynamic programming operator for the value iteration algorithm.

Our work is the first step towards many practical algorithms for solving discounted stochastic games with one-sided partial observability.
These can be applied in many scenarios requiring robust strategies and thus our work opens the whole new area of research in algorithmic and computational game theory.
One natural continuation is an adaptation of point-based approximation algorithms for POMDPs to improve the scalability of the value iteration algorithm.

\bibliographystyle{splncs03}
\bibliography{main}

\newpage
\appendix
\section*{APPENDIX}
\section{Proofs}

\subsection{Proof of Lemma~\ref{lemma:root_value}}
The expected reward in the root node $\uh(\emptyset)$ is
{\small
\begin{align*}
\uh(\emptyset) = \underbrace{\overbrace{\sum_{s_e^0 \in s_p^0} b(s_e^0)}^{\text{Evader is caught}} + \overbrace{\gamma \sum_{s_e^0 \not\in s_p^0} \sum_{s_p^1} \sum_{s_e^1 \in s_p^1} b(s_e^0) \cdot \pi_p(s_p^1) \cdot \pi_e(s_e^0,s_e^1)}^{\text{Evader is caught in the first round}}}_{\text{All terminal histories for the first round, i.e. shorter than two actions}} \ + \\
+ \underbrace{\sum_{s_e^0 \not\in s_p^0} \sum_{s_p^1} \sum_{s_e^1 \not\in s_p^1} b(s_e^0) \cdot \pi_p(s_p^1) \cdot \pi_e(s_e^0,s_e^1) \cdot \uh(s_e^0 s_p^1 s_e^1)}_{\text{Evader is not caught in the first round}}
\end{align*}
}
The derivation of Equation~\eqref{eq:root_value} is then just a technical derivation involving operations with sums and normalization of conditional probability distributions. Throughout the derivation we will use the equation of the transformed belief (Equation~\eqref{eq:belief}).
{\small
\begin{align*}
\uh(\emptyset) &= \sum_{s_e^0 \in s_p^0} b(s_e^0) + \gamma \sum_{s_e^0 \not\in s_p^0} \sum_{s_p^1} \sum_{s_e^1 \in s_p^1} b(s_e^0) \cdot \pi_p(s_p^1) \cdot \pi_e(s_e^0,s_e^1) \ \ + \\
& \qquad + \sum_{s_e^0 \not\in s_p^0} \sum_{s_p^1} \sum_{s_e^1 \not\in s_p^1} b(s_e^0) \cdot \pi_p(s_p^1) \cdot \pi_e(s_e^0,s_e^1) \cdot \uh(s_e^0 s_p^1 s_e^1) \\
&= \sum_{s_e^0 \in s_p^0} b(s_e^0) + \gamma \sum_{s_p^1} \pi_p(s_p^1) \sum_{s_e^1 \in s_p^1} \sum_{s_e^0 \not\in s_p^0} b(s_e^0) \cdot \pi_e(s_e^0,s_e^1) \ \ + \\
& \qquad + \sum_{s_p^1} \pi_p(s_p^1) \sum_{s_e^1 \not\in s_p^1} \sum_{s_e^0 \not\in s_p^0} b(s_e^0) \cdot \pi_e(s_e^0,s_e^1) \cdot \uh(s_e^0 s_p^1 s_e^1) \\
\end{align*}

\vspace{-5em}

\begin{align*}
&= \sum_{s_e^0 \in s_p^0} b(s_e^0) + \gamma \left[ \sum_{s_e^0 \not\in s_p^0} b(s_e^0) \right] \sum_{s_p^1} \pi_p(s_p^1) \sum_{s_e^1 \in s_p^1} \frac{\sum\limits_{s_e^0 \not\in s_p^0} b(s_e^0) \cdot \pi_e(s_e^0,s_e^1)}{\sum\limits_{s_e^0 \not\in s_p^0} b(s_e^0)} \ \ + \\
& + \sum_{s_p^1} \pi_p(s_p^1) \cdot \left[ \sum\limits_{\tilde{s}_e^1 \not\in s_p^1} \sum\limits_{\tilde{s}_e^0 \not\in s_p^0} b(s_e^0) \cdot \pi_e(s_e^0,s_e^1) \right] \sum_{s_e^1 \not\in s_p^1} \sum\limits_{s_e^0 \not\in s_p^0} \left[ \frac{b(s_e^0) \cdot \pi_e(s_e^0,s_e^1)}{\sum\limits_{\tilde{s}_e^1 \not\in s_p^1} \sum\limits_{\tilde{s}_e^0 \not\in s_p^0} b(s_e^0) \cdot \pi_e(s_e^0,s_e^1)} \cdot \uh(s_e^0 s_p^1 s_e^1) \right] \\
\end{align*}
\hfill

\begin{align*}
&= \sum_{s_e^0 \in s_p^0} b(s_e^0) + \left[ \sum_{s_e^0 \not\in s_p^0} b(s_e^0) \right] \Vast( \gamma \sum_{s_p^1} \pi_p(s_p^1) \sum_{s_e^1 \in s_p^1} \frac{\sum\limits_{s_e^0 \not\in s_p^0} b(s_e^0) \cdot \pi_e(s_e^0,s_e^1)}{\sum\limits_{s_e^0 \not\in s_p^0} b(s_e^0)} \ \ + \\
& \ + \sum_{s_p^1} \pi_p(s_p^1) \left[ \sum\limits_{\tilde{s}_e^1 \not\in s_p^1} \frac{\sum\limits_{\tilde{s}_e^0 \not\in s_p^0} b(s_e^0) \cdot \pi_e(s_e^0,s_e^1)}{\sum\limits_{s_e^0 \not\in s_p^0} b(s_e^0)} \right] 
\sum_{s_e^1 \not\in s_p^1} \sum\limits_{s_e^0 \not\in s_p^0} \left[ \frac{b(s_e^0) \cdot \pi_e(s_e^0,s_e^1)}{\sum\limits_{\tilde{s}_e^1 \not\in s_p^1} \sum\limits_{\tilde{s}_e^0 \not\in s_p^0} b(s_e^0) \cdot \pi_e(s_e^0,s_e^1)} \cdot \uh(s_e^0 s_p^1 s_e^1) \right] \Vast) \\
&= \sum_{s_e^0 \in s_p^0} b(s_e^0) + \left[ \sum_{s_e^0 \not\in s_p^0} b(s_e^0) \right] \Vast( \gamma \sum_{s_p^1} \pi_p(s_p^1) \sum_{s_e^1 \in s_p^1} b_{\pi_e}(s_e^1) \ \ + \\
& \ + \sum_{s_p^1} \pi_p(s_p^1) \left[ \sum\limits_{s_e^1 \not\in s_p^1} b_{\pi_e}(s_e^1) \right] \sum_{s_e^1 \not\in s_p^1} \sum\limits_{s_e^0 \not\in s_p^0} \left[ \frac{b(s_e^0) \cdot \pi_e(s_e^0,s_e^1)}{\sum\limits_{\tilde{s}_e^1 \not\in s_p^1} \sum\limits_{\tilde{s}_e^0 \not\in s_p^0} b(\tilde{s}_e^0) \cdot \pi_e(\tilde{s}_e^0,\tilde{s}_e^1)} \cdot \uh(s_e^0 s_p^1 s_e^1) \right] \Vast) \\
&= \sum_{s_e^0 \in s_p^0} b(s_e^0) + \left[ \sum_{s_e^0 \not\in s_p^0} b(s_e^0) \right] \cdot \sum_{s_p^1} \pi_p(s_p^1) \Vast( \gamma \sum_{s_e^1 \in s_p^1} b_{\pi_e}(s_e^1) \ \ + \\
& \quad + \left[ \sum\limits_{s_e^1 \not\in s_p^1} b_{\pi_e}(s_e^1) \right] \sum_{s_e^1 \not\in s_p^1} \sum\limits_{s_e^0 \not\in s_p^0} \left[ \frac{b(s_e^0) \cdot \pi_e(s_e^0,s_e^1)}{\sum\limits_{\tilde{s}_e^1 \not\in s_p^1} \sum\limits_{\tilde{s}_e^0 \not\in s_p^0} b(\tilde{s}_e^0) \cdot \pi_e(\tilde{s}_e^0,\tilde{s}_e^1)} \cdot \uh(s_e^0 s_p^1 s_e^1) \right] \Vast) \\
\end{align*}
}

\subsection{Proof of Lemma~\ref{lemma:subtree-elimination}}
\begin{proof}
Let $\calH_1$ be histories in the subtree of $n_1$, $\calH_2$ be histories in the subtree of $n_2$ and $\xi: \calH_1 \rightarrow \calH_2$ be a bijection from the definition of the game tree isomorphism.
Let $\sigma$ be a Nash equilibrium strategy profile in the game $G$ and $\sigma[n_1]$, $\sigma[n_2]$ be the optimal behaviors in subtrees of $n_1$ and $n_2$ when the strategy profile $\sigma$ is followed.

First of all, we show that we can modify the behavior in the subtree of $n_2$ to $\sigma'[n_2]$, where $\sigma'[n_2](n) = \sigma[n_1](\xi(n))$ without changing the expected utility (and hence we can assume that the strategies in both subtrees are the same).
Note that strategy profile $\sigma'[n_2]$ is valid in the subtree of $n_2$ due to the fact that subtrees beneath $n_1$ and $n_2$ are isomorphic.
Moreover this change does not affect behavior in the rest of the game, as $\sigma[n_1]$ was consistent with this behavior.

Assume that the strategy profile $\sigma'[n_2]$ is \emph{not} optimal in the subtree of $n_2$.
Strategies $\sigma[n_1]$ and $\sigma'[n_2]$ induce the same distribution over leaf nodes in subtrees of $n_1$ and $n_2$ (up to the bijection $\xi$) and these leaf nodes have the same utility values.
Hence if one of the players wanted to deviate from $\sigma'[n_2]$ in the subtree of $n_2$, he would have wanted to do the same in the case of $\sigma[n_1]$ used in the subtree of $n_1$.
Thus the strategy $\sigma'[n_2]$ must be an optimal behavior in the subtree of $n_2$.

Strategies $\sigma[n_1]$ and $\sigma'[n_2]$ induce the same distribution over the leafs in the respective subtrees (given that the node $n_1$, resp. $n_2$, was reached).
We construct a game $G'$ from $G$ by modifying the probability of reaching nodes $n_1$ and $n_2$ (the belief in information set $I$ of $G'$ is denoted $b'$).
The probability of reaching one of the terminal histories $u \in \calH_1$ and $\xi(u)$ in $G'$ remains the same as in $G$ as long as $b'[I](n_1)+b'[I](n_2)=b[I](n_1)+b[I](n_2)$.
The modification of the belief thus does not change the optimal behavior of the players, and thus does not change the value of the game.

\end{proof}

\subsection{Proof of Lemma~\ref{lemma:Qt}}
Assume that the pursuer played action $s_p^1$ in the first round of the game (drawn from $\pi_p$). The game moves to a shorter horizon game $G^{t-1}\!\angled{s_p^1,b_{\pi_e}}$ with the belief updated according to Equation~\eqref{eq:belief} (where $\pi_e$ is evader's best response to $\pi_p$). The optimal strategy $\sigma_p[s_p^1]$ of the pursuer in $G^{t-1}\!\angled{s_p^1,b_{\pi_e}}$ is optimal at belief $b_{\pi_e}$, hence, by definition, its corresponding $\alpha$-vector $\alpha[s_p^1]$ is present in the value function $v^{t-1}\!\angled{s_p^1}$ which expresses the value of all non-dominated strategies in the game $G^{t-1}\!\angled{s_p^1,b_{\pi_e}}$.

The value of every strategy is linear in belief (Lemma~\ref{lemma:v:mixed-strategy-linear}). It is therefore sufficient to define the value of the strategy $\sigma_p$ in each of the pure beliefs to form the $\alpha$-vector $\alpha^{\sigma_p}$. If the evader is located in the same vertex as one of the units of the pursuer, the game ends immediately and the pursuer gets utility $\gamma^0=1$. If the evader is not immediately caught, she chooses a vertex $s_e'$ adjacent to her current position $s_e$ so that the expected utility of the pursuer is minimized. We know that if the pursuer plays $s_p^1$ in the first round and follows with $\sigma_p[s_p^1]$ afterwards, the expected utility of the pursuer is represented by the $\alpha$-vector $\alpha_p[s_p^1]$ evaluated at the pure belief corresponding to the new position of the evader $s_e'$, multiplied by $\gamma$ as one round has already passed.

\subsection{Proof of Lemma~\ref{lemma:belief-finite}}
We show that if there is a belief point $b$, where the value function $\hat{v}^t\!\angled{s_p^0}$ can be improved, the value function $v^t\!\angled{s_p^0}$ can also be improved in at least one of the vertices of the facet on the surface of $v^t\!\angled{s_p^0}$.
We prove this by contradiction.
Let $F$ be the facet on the surface of $\hat{v}^t\!\angled{s_p^0}$ that the belief $b$ is projected on (i.e. $b$ is a convex combination of coordinates of vertices of $F$).
Assume that the value function $\hat{v}^t\!\angled{s_p^0}$ can be improved in $b$, but not in \emph{any} of the vertices of $F$.
This means that for every belief $b'$ corresponding to the extreme point of the facet $F$, the value function $\hat{v}^t\!\angled{s_p^0}$ is optimal, i.e. $\hat{v}^t\!\angled{s_p^0}(b') = v^t\!\angled{s_p^0}(b')$.
Moreover, as the value function can be improved at belief $b$, it holds that $\hat{v}^t\!\angled{s_p^0}(b) < v^t\!\angled{s_p^0}(b)$.
This means that the value of $v^t\!\angled{s_p^0}$ at $b$ is above facet $F$, as the original value $\hat{v}^t\!\angled{s_p^0}(b)$ was a convex combination of its vertices, which contradicts the convexity of the value function (Theorem~\ref{thm:convexity}).

\subsection{Proof of Lemma~\ref{lemma:contractivity}}
Let us define $Q^v_{\pi_p}\!\angled{s_p}$ similarly to $Q^t_{\pi_p}\!\angled{s_p}$ used in Section~\ref{sec:computing}, with the exception that we refer to an arbitrary value functions $v$ instead of $v^{t-1}$:
{\small
\begin{equation}
Q^v_{\pi_p}\!\angled{s_p}(b) \coloneqq \sum_{s_e \in s_p} b(s_e) + \gamma \left[ \sum_{s_e \in \calV \setminus s_p} b(s_e) \right] \cdot \min_{\pi_e} \sum_{s_p' \in \calV^N} \pi_p(s_p') \cdot v\angled{s_p'}(b_{\pi_e})
\end{equation}}
The proof will closely follow the structure of the proof of Theorem~1 in~\cite{smith2012point}. First of all we show that for every $s_p \in \calV^N$ and every valid pursuer's one-step strategy $\pi_p$, the mapping $v \mapsto Q^v_{\pi_p}\!\angled{s_p}$ has a contractivity factor $\gamma$, then by inspecting all possible $s_p$ and $\pi_p$ we show that the same holds for $H$.
Note that the difference $| Q^{v}_{\pi_p}\!\angled{s_p}(b) - Q^{\overline{v}}_{\pi_p}\!\angled{s_p}(b) |$ is maximized if $b$ is chosen so that the evader is not initially caught according to $b$ (i.e. $\sum_{s_e \in s_p} b(s_e) = 0$), which allows us to simplify the derivation.
{\small \begin{align*}
\left\| Q^{v}_{\pi_p}\!\angled{s_p} - Q^{\overline{v}}_{\pi_p}\!\angled{s_p} \right\| &= \max_b \Big| Q^{v}_{\pi_p}\!\angled{s_p}(b) - Q^{\overline{v}}_{\pi_p}\!\angled{s_p}(b) \Big| \\
&\hspace{-80pt} = \max_b \left| \gamma \min_{\pi_e} \textstyle\sum_{s_p'} \pi_p(s_p') \cdot v\angled{s_p'}(b_{\pi_e}) - \gamma \displaystyle\min_{\pi_e'} \textstyle\sum_{s_p'} \pi_p(s_p') \cdot \overline{v}\angled{s_p'}(b_{\pi_e'}) \right| \\
&\hspace{-80pt} = \gamma \max_b \left| \min_{\pi_e} \textstyle\sum_{s_p'} \pi_p(s_p') \cdot v\angled{s_p'}(b_{\pi_e}) - \displaystyle\min_{\pi_e'} \textstyle\sum_{s_p'} \pi_p(s_p') \cdot \overline{v}\angled{s_p'}(b_{\pi_e'}) \right| \\
&\hspace{-80pt} \leq \gamma \max_b \max_{\pi_e} \left| \textstyle\sum_{s_p'} \pi_p(s_p') \cdot v\angled{s_p'}(b_{\pi_e}) - \sum_{s_p'} \pi_p(s_p') \cdot \overline{v}\angled{s_p'}(b_{\pi_e}) \right| \\
&\hspace{-80pt} \leq \gamma \max_{b'} \left| \textstyle\sum_{s_p'} \pi_p(s_p') \cdot \left[ v\angled{s_p'}(b') - \overline{v}\angled{s_p'}(b') \right] \right| \\
&\hspace{-80pt} \leq \gamma \max_{b'} \textstyle\sum_{s_p'} \pi_p(s_p') \cdot \left| v\angled{s_p'}(b') - \overline{v}\angled{s_p'}(b') \right| \\
&\hspace{-80pt} \leq \gamma \max_{b'} \textstyle\sum_{s_p'} \pi_p(s_p') \cdot \| v - \overline{v} \| \\
&\hspace{-80pt} \leq \gamma \| v - \overline{v} \| \\
\end{align*}}

Let us now choose arbitrary $s_p \in \calV^N$ and $b \in \Delta(\calV)$. Without loss of generality, let us assume that $(Hv)\angled{s_p}(b) \geq (H\overline{v})\angled{s_p}(b)$. Let $\pi_p^*$ be an optimal one-shot strategy in $b$ w.r.t. $v$ (i.e. maximizing $Q^v_{\pi_p^*}\angled{s_p}(b)$) and $\overline{\pi}_p$ an optimal strategy w.r.t. $\overline{v}$. Note that $(Hv)\angled{s_p}(b)=Q^v_{\pi_p^*}\angled{s_p}(b)$ and $(H\overline{v})\angled{s_p}(b)=Q^{\overline{v}}_{\overline{\pi}_p}\angled{s_p}(b)$. It holds that $Q^{\overline{v}}_{\pi_p^*}\angled{s_p}(b) \leq Q^{\overline{v}}_{\overline{\pi}_p}\angled{s_p}(b) \leq Q^v_{\pi_p^*}\angled{s_p}(b)$. Then:
{\small \begin{align*}
\left| (Hv)\angled{s_p}(b) - (H\overline{v})\angled{s_p}(b) \right| &= \left| Q^v_{\pi_p^*}\angled{s_p}(b) - Q^{\overline{v}}_{\overline{\pi}_p}\angled{s_p}(b) \right| \\
&\leq \left| Q^v_{\pi_p^*}\angled{s_p}(b) - Q^{\overline{v}}_{\pi_p^*}\angled{s_p}(b) \right| \\
&\leq \max_{\pi_p} \left| Q^v_{\pi_p}\angled{s_p}(b) - Q^{\overline{v}}_{\pi_p}\angled{s_p}(b) \right| \\
&\leq \gamma \cdot \left\| v - \overline{v} \right\|
\end{align*}}

\subsection{Proof of Proposition~\ref{prop:convergence}}
The minimum reward the pursuer can get is zero, hence $v^0\!\angled{s_p^0}(b) \geq 0$ for every initial position of the pursuer $s_p^0$ and every his belief $b$ about the position of the evader. Similarly the maximum reward is one, hence $v^*\!\angled{s_p^0}(b) \leq 1$ for every $s_p^0$ and every $b$. It holds therefore that $\| v^0 - v^* \| \leq 1$. Due to the contractivity factor $\gamma$ of the dynamic operator $H$, after $t$ iterations the distance $\| v^t - v^* \| \leq \gamma^t \| v^0 - v^* \| \leq \gamma^t$ which completes the proof.

\section{Computing optimal strategies}
\subsection{Computing evader's strategy}
The idea behind the linear program for computing optimal evader's one-step strategy in game $G^t\!\angled{s_p^0,b}$ is similar to the one for solving the game from the perspective of the pursuer (Section~\ref{sec:lp}).
In this case the roles of the players are reversed --- the evader seeks for her strategy, while the pursuer best-responds it.
Similarly as in Section~\ref{sec:lp}, we focus on the case where $b(s_e^0)=0$ for all $s_e^0 \in s_p^0$.

By choosing the strategy $\pi_e$, the evader decides the belief $b_{\pi_e}$ in the shorter horizon game $G^{t-1}\!\angled{s_p^1,b_{\pi_e}}$ (see Equation~\eqref{eq:belief}).
The calculation of this updated belief is done by means of the Constraint~\eqref{lp:evader:belief}.
If the pursuer decides to move to the set of vertices $s_p^1$, his expected utility is described by the value function $v^{t-1}\!\angled{s_p^1}$ multiplied by $\gamma$ to account for the first round of the game. 
The pursuer best-responds by choosing the best $s_p^1$ from his perspective, i.e. the one that maximizes his utility, which is characterized by a set of best-response constraints (Equation~\eqref{lp:evader:br}).
Note that the maximization over $\alpha$-vectors in the value functions is rewritten using a set of inequality constraints, one for each $\alpha$-vector.
The evader then seeks for a strategy that minimizes the expected utility of the pursuer (Equation~\eqref{lp:evader:objective}).

{\small
\begin{align}
\min_{V,\pi_e,b_{\pi_e}} & \ \ \ \ V \label{lp:evader:objective}\\
\text{s.t.}\; & \gamma \sum_{s_e' \in \calV} \alpha(s_e') \cdot b_{\pi_e}(s_e') \leq V && \forall s_p^1 \in(s_p^0) \ \ \forall \alpha \in v^{t-1}\!\angled{s_p^1} \label{lp:evader:br}\\
& \sum_{s_e \in \calV \setminus s_p^0} \!\!\!\! b(s_e) \cdot \pi_e(s_e,s_e') = b_{\pi_e}(s_e') && \forall s_e' \in \calV \label{lp:evader:belief}\\
& \! \sum_{s_e' \in \adj(s_e)} \!\!\!\!\!\!\!\! \pi_e(s_e,s_e') = 1 && \forall s_e \in \calV\\
& \ \ \ \ \pi_e(s_e,s_e') \geq 0 && \forall s_e \in \calV \ \ \forall s_e' \in \adj(s_e)
\end{align}
}

\end{document}